\documentclass[11pt,a4paper, final, twoside]{article}
\usepackage{amsmath}
\usepackage{fancyhdr}
\usepackage{amsthm}
\usepackage{amsfonts}
\usepackage{amssymb}
\usepackage{amscd}
\usepackage{mathabx}
\usepackage{latexsym}
\usepackage{amsthm}
\usepackage{graphicx}
\usepackage{graphics}
\usepackage{afterpage}
\usepackage[colorlinks=true, urlcolor=blue,  linkcolor=black, citecolor=black]{hyperref}
\usepackage{color}
\newtheorem{thm}{thm}[section]

\theoremstyle{proposition}

\theoremstyle{definition}

\theoremstyle{rem}
\newtheorem{rem}{rem}[section]

\numberwithin{equation}{section}
\newcommand{\RR}{\mathbb{R}}
\newcommand{\NN}{\mathbb{N}}

\begin{document}
\hyphenpenalty=100000

\begin{flushright}

{\Large \textbf{\\Relaxing the Hypotheses of Symmetry and Time-Reversibility in Genome Evolutionary Models}}\\[5mm]
{\large \textbf{Jacques M. Bahi$^\mathrm{1}$, Christophe Guyeux$^\mathrm{1}$} and \textbf{Antoine Perasso$^\mathrm{*2}$ \footnote{\emph{*Corresponding author: E-mail: antoine.perasso@univ-fcomte.fr}}}}\\[1mm]
$^\mathrm{1}${\footnotesize \it UMR6147~Institut FEMTO-ST, Universit\'e de Franche-Comt\'e, 25030 Besan\c con, France}\\ $^\mathrm{2}${\footnotesize \it UMR6249~Chrono-environnement, Universit\'e de Franche-Comt\'e, 25030 Besan\c con, France}
\end{flushright}

{\Large \textbf{Abstract}}\\[4mm]
\fbox{%
\begin{minipage}{5.4in}{\footnotesize Various genome evolutionary models have been proposed these last decades to predict the evolution of a DNA sequence over time, essentially described using a mutation matrix. By essence, all of these models
relate the evolution of DNA sequences to the computation of the successive powers of the mutation matrix. To make this computation possible, hypotheses are assumed for the matrix, such as symmetry and time-reversibility, which are not compatible with mutation rates that have been recently obtained experimentally on genes $ura3$ and $can1$ of the Yeast 
\textit{Saccharomyces cerevisiae}. In this work, authors investigate systematically the possibility to relax either the symmetry or the time-reversibility hypothesis
 of the mutation matrix, by investigating all the possible matrices 
of size $2\times 2$ and $3 \times 3$.
As an application example, the experimental study on the Yeast \textit{Saccharomyces cerevisiae}
has been used in order to deduce a simple mutation matrix, and to compute the future evolution of the rate 
purine/pyrimidine for $ura3$ on the one hand, and of the particular behavior of cytosines and thymines compared to purines 
on the other hand.
} \end{minipage}}\\[1mm]

\afterpage{
\fancyhead{} \fancyfoot{} 
\fancyfoot[R]{\footnotesize\thepage}
\fancyhead[R]{\scriptsize\it British Journal of Mathematics and Computer Science 
{{X(X), XX--XX}},~20XX \\
 }}

\section{Introduction}

Due to mutations or recombination, some variations occur 
in the frequency of each codon, and these codons are thus not uniformly distributed into a given genome.
Since the late `60s, various genome evolutionary models have been proposed to
predict the evolution of a DNA sequence as generations pass. 
Mathematical models allow the prediction of such an evolution, in such a way 
that statistical values observed in current genomes can be at least partially recovered
from hypotheses on past DNA sequences.
Moreover, it can be attractive to study the genetic patterns (blocs of more than one 
nucleotide: dinucleotides, trinucleotides...) that appear and disappear depending on mutation parameters.

A first model for genomes evolution has been proposed in 1969 by Thomas 
Jukes and Charles Cantor \cite{Jukes69}. This first model is very simple,
as it supposes that each nucleotide has the probability $m$ to
mutate to any other nucleotide, as described in the following mutation
matrix,
$$
\left(\begin{array}{cccc}
* & m & m & m\\
m & * & m & m\\
m & m & * & m\\
m & m & m & *\\
\end{array}\right).
$$
In that matrix, the nucleotides are ordered as $(A,C,G,T)$, so that for instance the coefficient in row 3, column 2 represents the 
probability that the nucleotide $G$ mutates into a $C$ during the next time
interval, \emph{i.e.}, $P(G \rightarrow C)$. As diagonal elements can be 
deduced by the fact that the sum of each row must be equal to 1, they are
omitted here.

This first attempt has been followed up by Motoo Kimura \cite{Kimura80}, 
who has reasonably considered that transitions ($A \longleftrightarrow G$ and 
$T \longleftrightarrow C$) should not have the same mutation rate than
transversions ($A \longleftrightarrow T$, $A \longleftrightarrow C$, $T \longleftrightarrow G$, and $C \longleftrightarrow G$), 
this model being refined by Kimura in 1981, with three 
constant parameters to make a distinction between natural 
$A\longleftrightarrow T$, $C \longleftrightarrow G$ and unnatural
transversions, leading to: 
$$
\left(\begin{array}{cccc}
* & c & a & b\\
c & * & b & a\\
a & b & * & c\\
b & a & c & *\\
\end{array}\right).
$$

Joseph Felsenstein \cite{Felsenstein1980} has then supposed that the nucleotides 
frequency depends on the kind of nucleotide A,C,T,G. Such a supposition leads
to a mutation matrix of the form:
$$
\left(\begin{array}{cccc}
* & \pi_C & \pi_G & \pi_T\\
\pi_A & * & \pi_G & \pi_T\\
\pi_A & \pi_C  & * & \pi_T\\
\pi_A & \pi_C  & \pi_G & *\\
\end{array}\right)
$$
with $\pi_A$, $\pi_C$, $\pi_G$, and $\pi_T$ denoting the frequency of occurance of each nucleotide, respectively.
Masami Hasegawa, 
Hirohisa Kishino, and Taka-Aki Yano \cite{Hasegawa1985} have generalized the 
models of \cite{Kimura80} and \cite{Felsenstein1980}, introducing in 1985 the
following mutation matrix:
$$
\left(\begin{array}{cccc}
* & \alpha \pi_C & \beta \pi_G & \alpha \pi_T\\
\alpha \pi_A               & * &\alpha \pi_G & \beta\pi_T\\
\beta \pi_A               & \alpha\pi_C  & * & \alpha\pi_T\\
\alpha \pi_A               & \beta\pi_C  & \alpha \pi_G & *\\
\end{array}\right).
$$

These efforts have been continued by Tamura, who proposed in~\cite{Tamura92,Tamura93}
a simple method to estimate the number of nucleotide substitutions per site 
between two DNA sequences, by extending the model of Kimura
(1980). The idea is to consider a two-parameter method, for the case where a 
GC bias exists. Let us denote by $\pi_{GC}$ the frequency of this dinucleotide
motif. Tamura supposes that $\pi_G = \pi_C = \dfrac{\pi_{GC}}{2}$ and
$\pi_A = \pi_T = \dfrac{1-\pi_{GC}}{2}$, which leads to the following rate matrix: 
$$\begin{pmatrix} {*} & {\kappa(1-\pi_{GC})/2} & {(1-\pi_{GC})/2} & {(1-\pi_{GC})/2} \\ {\kappa\pi_{GC}/2} & {*} & {\pi_{GC}/2} & {\pi_{GC}/2} \\ {(1-\pi_{GC})/2} & {(1-\pi_{GC})/2} & {*} & {\kappa(1-\pi_{GC})/2} \\ {\pi_{GC}/2} & {\pi_{GC}/2} & {\kappa\pi_{GC}/2} & {*} \end{pmatrix}.$$

All these models are special cases of the GTR model~\cite{yang94}, in which the
mutation matrix has the form (using obvious notations):
$$
\left(\begin{array}{cccc}
* & f_{AC} \pi_C & f_{AG} \pi_G & f_{AT} \pi_T\\
f_{AC} \pi_A               & * & f_{CG} \pi_G & f_{CT} \pi_T\\
f_{AG} \pi_A               & f_{CG} \pi_C  & * & \pi_T\\
f_{AT} \pi_A               & f_{CT}\pi_C  & \pi_G & *\\
\end{array}\right).
$$
Non-reversible and non-symmetric models have, for their part, been considered in practical inferences since
at least a decade for phylogenetic studies, see for instance~\cite{Klosterman,Boussau,Yap}
Furthermore, in the nonlinear case, a mutation can exhibit a chaos, see for instance~\cite{Ganikhodjaev}.
As they are more regarded for their interest in practical inference investigations than on the theoretical side, they will not be developed in this article.

\begin{table}
\centering
\scriptsize
 \begin{tabular}{lcc}
  \hline
Mutation & $ura3$ & $can1$ \\
\hline
$T \rightarrow C$ & 4 & 4\\
$T \rightarrow A$ & 14 & 9\\
$T \rightarrow G$ & 5 & 5\\
$C \rightarrow T$ & 16 & 20\\
$C \rightarrow A$ & 40 & 21\\
$C \rightarrow G$ & 11 & 9\\
$A \rightarrow T$ & 8 & 4\\
$A \rightarrow C$ & 6 & 5\\
$A \rightarrow G$ & 0 & 1\\
$G \rightarrow T$ & 28 & 20\\
$G \rightarrow C$ & 9 & 12\\
$G \rightarrow A$ & 26 & 40\\
Transitions & 46 & 65\\
Transversions & 121 & 85\\
\hline
 \end{tabular}
\caption{\small Summary of sequenced $ura3$ and $can1$ mutations~\cite{Lang08}}
\label{ura3taux}
\end{table}

Due to mathematical complexity~\cite{Knopoff1,Knopoff2}, matrices theoretically investigated to model evolution of
DNA sequences are thus limited either by the hypotheses of symmetry and time-reversibility or by
the desire to reduce the number of parameters under consideration. These hypotheses
allow their authors to solve theoretically the DNA evolution problem, for instance by 
computing directly the successive powers of their mutation matrix. However, one
can wonder whether such restrictions on the mutation rates are realistic. Focusing
on this question, we used in~\cite{bgp12:ip} a recent research work of Lang and Murray~\cite{Lang08},
in which the per-base-pair mutation rates of the Yeast \textit{Saccharomyces cerevisiae}
have been experimentally measured (see Table~\ref{ura3taux}), allowing us to calculate concrete mutation matrices for genes $ura3$ and $can1$.
%
We deduced in \cite{bgp12:ip} that none 
of the existing genomes evolution models can fit such mutation
matrices, implying the fact that some hypotheses must be relaxed, even if this
relaxation implies less ambitious models: current models do not
match with what really occurs in concrete genomes, at least in the case of this 
yeast.
Having these considerations in mind, the data obtained by Lang and Murray
have been used in~\cite{bgp12:ip} in order to predict the evolution of 
the rates or purines and pyrimidines in the particular case of $ura3$.
Mathematical investigations and numerical simulations have been proposed, 
focusing on this particular gene and its associated matrix of size $2 \times 2$
(purines \emph{vs.} pyrimidines), and of size $3 \times 3$ (cytosines and thymines compared to purines).
Instead of focusing on two particular matrices, this extension of~\cite{bgp12:ip} 
investigates systematically all the possible mutation matrices of sizes $2 \times 2$
and $3 \times 3$. Thus, the study is finalized in this article, by investigating
all the possible cases, and discussing about their mathematical and biological 
relevance.

The remainder of this research work is organized as follows.
First of all the case of mutation matrices
of size $2 \times 2$ is recalled in Section \ref{Model22}  and applied to the $ura3$ gene 
taking into account purines and pyrimidines mutations. A simulation is then performed to compare this non reversible model to the classical symmetric Cantor model. 
The next sections deal with all the possible 6-parameters  models of size $3 \times 3$.
In Section \ref{Model33}, a complete theoretical study is led encompassing all the particular situations, 
whereas in Section~\ref{appnum} an illustrative example focusing on the  evolution of the
purines, cytosines, and thymines triplet is given for $ura3$. 
We finally conclude this work in Section \ref{Conclusion}.





\section{General Model of Size $2\times 2$}
\label{Model22}

In this section, a first general genome evolution model focusing on purines
versus pyrimidines is proposed, to illustrate the method and as a pattern
for further investigations. This model is applied to the case of the yeast
\textit{Saccharomyces cerevisiae}.

\subsection{A convergence result}

Let $R$ and $Y$ denote respectively the occurrence frequency of purines and pyrimidines in a sequence
of nucleotides, and $M=\left(
\begin{array}{cc}
	a & b \\
	c & d
\end{array}\right)
$
the associated mutation matrix, with
$a = P(R\to R)$, $b=P(R\to Y)$, $c=P(Y\to R)$, and $d=P(Y\to Y)$
satisfying
\begin{equation}\label{stochastic}
\begin{cases}
a+b = 1, \\
c+d = 1, \\
\end{cases}
\end{equation}
and thus $M=\left(
\begin{array}{cc}
	a & 1-a \\
	c & 1-c
\end{array}\right)
$.

The initial probability is denoted by $P_0 = (R_0 ~~ Y_0)$, where $R_0$ and $Y_0$ denote
respectively the initial frequency of purines and pyrimidines. So the occurrence 
probability at generation $n$ is $P_n =  P_0 M^n$, where $P_n=(R(n) ~~ Y(n))$ 
is a probability vector such that $R(n)$ (resp. $Y(n)$) is the rate of purines
(resp. pyrimidines) after $n$ generations.
The following theorem states the time asymptotic behavior of the probabilit $P_n$.

We recall the following result was proved in \cite{bgp12:ip}:
\begin{thm}
\label{th2d}
Consider a DNA sequence under evolution, whose mutation matrix is
$M=\left(
\begin{array}{cc}
	a & 1-a \\
	c & 1-c
\end{array}\right)
$
with $a=P(R\to R)$ and $c = P(Y\to R)$. 
\begin{itemize}
 \item If $a=1, c=0$, then the frequencies of purines and pyrimidines do not 
change as the generation pass.
 \item If $a=0, c=1$, then these frequencies oscillate at each generation between 
$(R_0 ~~ Y_0)$ (even generations) and $(Y_0 ~~ R_0)$ (odd generations).
 \item Else the value $P_n = (R(n) ~~ Y(n))$ of purines and pyrimidines 
frequencies at generation $n$ is convergent to the following limit:
\begin{equation*}
\lim_{n\to\infty} P_n  = \frac{1}{c+1-a}  \left(
\begin{array}{cc}
	c &	1-a
\end{array}\right).
\end{equation*}
\end{itemize}
\end{thm}

\begin{rem}
Note that the case $a\neq 1-c$, resp. $a\neq c$, translates the non symmetry property, resp. the time reversibility property.
\end{rem}
\begin{proof}
To prove the theorem, we have to determinate $M^n$ for evey $n\in \NN$.\\
A division algorithm leads to the existence of a polynomial of degree $n-2$,
denoted by $Q_M \in \RR_{n-2}[X]$, and to $a_n,b_n\in\RR$ such that
\begin{equation}\label{div}
X^n = Q_M(X) \chi_M(X) + a_n X + b_n,
\end{equation}
when $\chi_M$ is the characteristic polynomial of $M$. 
Using both the Cayley-Hamilton theorem and the 
equality given above, we thus have
\begin{equation*}
M^n = a_n M + b_n I_2.
\end{equation*}

In order to determine $a_n$ and $b_n$, we must find the roots of $\chi_M$. As $\chi_M(X) = X^2- Tr(M) X + \det(M)$ and due to \eqref{stochastic},  we can conclude that $1$ is a root of $\chi_M$, which thus has two real roots: $1$ and $x_2$. 
As the roots sum is equal to -tr(A), we conclude that $x_2 = a-c$.

If $x_2 = a-c = 1$, then $a=1$ and $c=0$ (as these parameters are in $[0,1]$),
so the mutation matrix is the identity and the frequencies of purines and pyrimidines
into the DNA sequence does not evolve. If not, evaluating  \eqref{div} in both $X=1$ and $X=x_2$, we thus obtain
\begin{equation*}
\begin{cases}
1 = a_n + b_n, \\
(a-c)^n = a_n (a-c) + b_n.
\end{cases}
\end{equation*}
Considering that $a-c \neq 1$, we obtain
\begin{equation*}
a_n  = \frac{(a-c)^n - 1}{a-c-1}, \qquad b_n = \frac{a-c - (a-c)^n}{a-c-1}.
\end{equation*}
Using these last expressions into the equality 
linking $M$, $a_n$, and $b_n$, we thus deduce the value of $P_n = P_0 M^n$, where
\begin{equation}
\label{mn2d}
M^n = \frac{1}{a-c-1} \left(
\begin{array}{cc}
	(a-1) (a-c)^n - c & (1-a) ((a-c)^n -1)) \\
	c ((a-c)^n-1) & -c (a-c)^n + a -1
\end{array}\right).
\end{equation}

If $a=0$ and $c=1$, then 
$M = \left(\begin{array}{cc}0 & 1 \\ 1 & 0\end{array}\right)$, so $M^{2n}$ is
the identity $I_2$ whereas $M^{2n+1}$ is $M$. Contrarily, if 
$(a,c) \notin \{(0,1); (1,0)\}$, then the limit of $M^n$ can be easily found using \eqref{mn2d}.\\
All the studided cases for $M^n$ lead to Theorem \ref{th2d}.
\end{proof}

\subsection{Numerical Application}

For numerical application, we will consider mutations rates in the \emph{ura3} gene of the Yeast \emph{Saccharomyces cerevisiae}, as obtained by Gregory I. Lang and Andrew W.Murray~\cite{Lang08} and summed up in Table~\ref{ura3taux}. They have measured phenotypic mutation rates, indicating that the per-base pair mutation rate at \emph{ura3} is $m=3.0552\times 10^{-7}$/generation for the whole gene.

For the majority of Yeasts they studied, \emph{ura3} is constituted by 804 bp: 133 cytosines,
211 thymines, 246 adenines, and 214 guanines. So $R_0 = \dfrac{246+214}{804} \approx 0.572$,
and $Y_0 = \dfrac{133+211}{804} \approx 0.428$. 
Using these values in the historical model of Jukes and Cantor~\cite{Jukes69}, we obtain the evolution depicted in
Figure~\ref{Cantor2d}. 

\begin{figure}[h]
\centering
\includegraphics[scale=0.5]{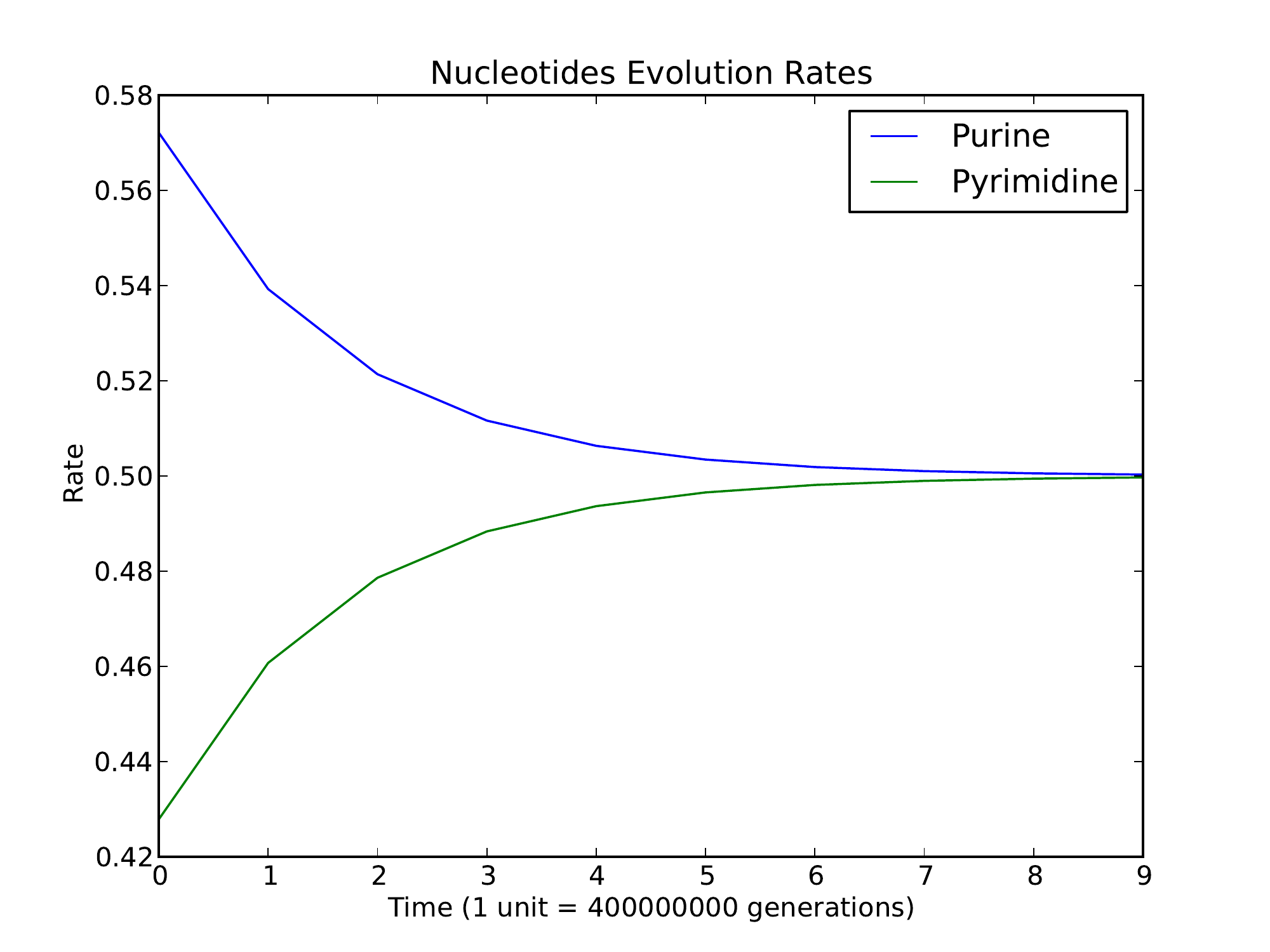}
\caption{\small Prediction of purine/pyrimidine evolution of $ura3$ gene in symmetric Cantor model.}
\label{Cantor2d}
\end{figure}

Theorem \ref{th2d} allows us to compute the limit of the rates of purines and 
pyrimidines:
\begin{description}
 \item[Computation of probability $a$.]
$P(R \rightarrow R) = (1-m)$ \linebreak $+P(A \rightarrow G)\dfrac{P_A(n)}{P_A(n)+P_G(n)} + P(G \rightarrow A)\dfrac{P_G(n)}{P_A(n)+P_G(n)}$. The use of Table~\ref{ura3taux} 
and the hypothesis that the base frequencies have already reached their
steady states implies
that $a = (1-m) + \left( m \dfrac{0}{46+121}\right) \times \dfrac{\frac{246}{804}}{\frac{246}{804}+\frac{214}{804}} +\left( m \dfrac{26}{46+121}\right) \times \dfrac{\frac{214}{804}}{\frac{246}{804}+\frac{214}{804}}$.
We thus obtain that $a=1-\dfrac{17814m}{19205}\approx 0.999999716$.
 \item[Computation of probability $c$.]
Similarly, $P(Y \rightarrow Y) = (1-m) + P(C \rightarrow T) \dfrac{P_C}{P_C+P_T}+
P(T \rightarrow C) \dfrac{P_T}{P_C+P_T}$ $=(1-m) + m \dfrac{16}{46+121} \times \dfrac{133}{133+211} + m \dfrac{4}{46+121} \times \dfrac{211}{133+211}$ $=1-m + m\dfrac{743}{14362}$.
So $c = P(Y \rightarrow R) 
= m\left(1-\dfrac{743}{14362}\right) \approx 2.897\times 10^{-7}$.
\end{description}

As a consequence the purine/pyrimidine mutation matrix that corresponds to the values of Table~\ref{ura3taux} is:
\begin{equation}\label{mutmat1}
M=m\left(
\begin{array}{cc}
	\dfrac{1391}{19205} & \dfrac{17814}{19205} \\\\
	\dfrac{13619}{14362} & \dfrac{743}{14362}
\end{array}\right).
\end{equation}
where $m=3.0552\times 10^{-7}$ as mentionned previously. 

Using the value of $m$ for the $ura3$ gene leads
to $1-a=2.83391\times 10^{-7}$ and $c=2.89714\times 10^{-7}$,
which can be used in Theorem~\ref{th2d} to 
conclude that the rate of pyrimidines is convergent
to $49.45\%$ whereas the rate of purines 
converge to $50.55\%$.
Numerical simulations using data published 
in~\cite{Lang08} are given in
Figure~\ref{evolution par notre methode}, leading
to a similar conclusion.

\begin{figure}
\centering
\includegraphics[scale=0.5]{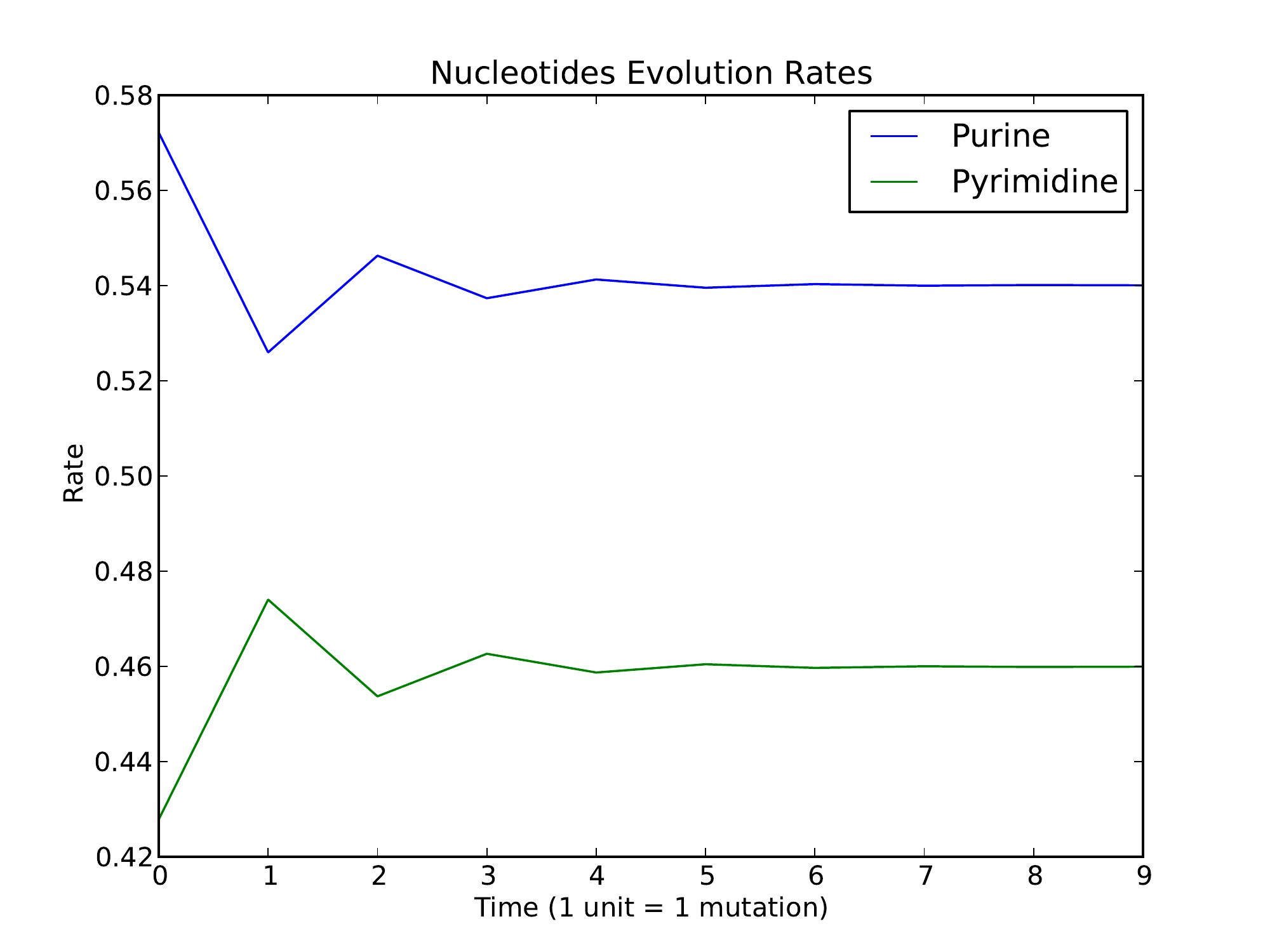}
\caption{\small Prediction of purine/pyrimidine evolution of $ura3$ gene in general model of size $2\times 2$.}
\label{evolution par notre methode}
\end{figure}

\section{A First Genomes Evolution Model of size $3 \times 3$ having 6 Parameters without Time-reversibility hypothesis}
\label{Model33}

In order to investigate the evolution of the frequencies of cytosines and thymines
in the gene $ura3$, a model of size $3 \times 3$ compatible with real
mutation rates of the yeast \textit{Saccharomyces cerevisiae} is now presented.

\subsection{Formalization}

Let us consider a line of yeasts where a given gene is sequenced at
each generation, in order to clarify explanations. The $n-$th generation
is obtained at time $n$, and the frequences of purines, cytosines, and thymines
at time $n$ are respectively denoted by $P_R(n), P_C(n)$, and $P_T(n)$.

Let $a$ be the probability that a purine is changed into a cytosine between
two generations, that is:  $a = P(R \rightarrow C)$. 
Similarly, denote by $b,c,d,e,f$ the 
respective probabilities: $P(R \rightarrow T)$, $P(C \rightarrow R)$, 
$P(C \rightarrow T)$, 
$P(T \rightarrow R)$, and $P(T \rightarrow C)$. Contrary to existing 
approaches, $P(R \rightarrow C)$ is not supposed to be equal to 
$P(C \rightarrow R)$, and the
same statement holds for the other probabilities. For the sake of simplicity,
we will suppose in all that follows that $a,b,c,d,e,f$ are not 
time dependent.

Let 
\begin{equation*}
 M = \left( 
\begin{array}{ccc}
1-a-b & a & b \\
c & 1-c-d & d \\
e & f & 1-e-f
\end{array}
\right)
\end{equation*}
be the mutation matrix associated to the probabilities mentioned above, and $P_n$ 
the vector of occurrence, at time $n$, of each of the three kind of
nucleotides. In other words, 
$P_n = (P_R(n) ~~ P_C(n) ~~ P_T(n))$. 
Under that hypothesis, $P_n$ is a 
probability vector: $\forall n \in \mathbb{N},$
\begin{itemize}
 \item $P_R(n), P_C(n), P_T(n) \in [0,1]$,
 \item $P_R(n) + P_C(n) + P_T(n) = 1$,
\end{itemize}
Let $P_0 = (P_R(0) ~~ P_C(0) ~~ P_T(0)) \in [0,1]^3$ be the initial probability vector. We have 
obviously: 
\begin{equation*}
P_R(n+1) = P_R(n) P(R \rightarrow R)+
P_C(n) P(C \rightarrow R)+P_T(n) P(T \rightarrow R), 
\end{equation*}
with similar equalities for $P_C(n+1)$ and $P_T(n+1)$ so that
\begin{equation}
\label{Pn}
 P_n = P_{n-1} M = P_0 M^n.
\end{equation}
In all that follows we wonder if, given the parameters $a,b,c,d,e,f$ as in~\cite{Lang08},
one can determine the frequency of occurrence of any of the three kind
of nucleotides when $n$ is sufficiently large, in other words if the limit of
$P_n$ is accessible by computations.

\subsection{Resolution}

This section, that is a preliminary of the convergence study, is devoted to the determination of the powers of matrix $M$ in the general case and some particular situations

\subsubsection{Determination of $M^n$ in the general case}

The characteristic polynomial of $M$ is equal to
\begin{equation*}
\begin{array}{cl} 
\chi_M(x) &= x^3 +(s-3) x^2+(p-2s+3)x-1+s-p \\
 & = (x-1)\left(x^2+(s-2)x+(1-s+p)\right),
\end{array}
\end{equation*}
where
\begin{gather*}
s = a+b+c+d+e+f, \\
p=ad+ae+af+bc+bd+bf+ce+cf+de, \\
det(M)  =1-s+p.
\end{gather*}

The discriminant of the polynomial of degree 2 in the factorization of $\chi_M$ is equal to 
$\Delta = (s-2)^2-4(1-s-p) = s^2-4p$. Let $x_1$ and
$x_2$ the two roots (potentially complex or 
equal) of $\chi_M$, given by
\begin{equation}
\label{x1x2}
x_1 = \dfrac{-s+2-\sqrt{s^2-4p}}{2} \textrm{ and }  x_2 = \dfrac{-s+2+\sqrt{s^2-4p}}{2}. 
\end{equation}

Let $n \in \mathbb{N}, n \geqslant 2$. As $\chi_M$ is a polynomial of degree 3, a division algorithm of $X^n$ by $\chi_M(X)$ leads to the
existence and uniqueness of two polynomials $Q_n$ and $R_n$, such that
\begin{equation}
\label{resol}
X^n = Q_n(X)\chi_2(X)+R_n(X),
\end{equation}
where the degree of $R_n$ is lower than or equal to
the degree of $\chi_M$, \emph{i.e.}, 
$R_n(X) = a_n X^2 + b_n X + c_n$ with $a_n, b_n, c_n \in \mathbb{R}$ for every $n\in\mathbb{N}$.
By evaluating \eqref{resol} in the three
roots of $\chi_M$, we find the system
$$\left\{ 
\begin{array}{cl}
1 & = a_n + b_n + c_n\\
x_1^n & = a_n x_1^2 + b_n x_1 + c_n\\
x_2^n & = a_n x_2^2 + b_n x_2 + c_n\\
\end{array}
\right.$$
This system is equivalent to
$$\left\{ 
\begin{array}{cccccl}
c_n & + &  b_n  &+ & a_n & = 1\\
    && b_n(x_1-1)& + & a_n(x_1^2-1) & = x_1^n-1\\
    && b_n(x_2-1) &+ & a_n(x_2^2-1) &  = x_2^n-1\\ 
\end{array}
\right.$$
If we suppose that $x_1 \neq 1$, $x_2 \neq 1$, and $x_1 \neq x_2$,
then standard algebraic computations give
\begin{equation*}
\left\{
\begin{array}{l}
a_n = \dfrac{1}{x_2-x_1}\left[\dfrac{x_2^n-1}{x_2-1}-\dfrac{x_1^n-1}{x_1-1}\right],\\
\\
b_n = \dfrac{x_1+1}{x_1-x_2}\dfrac{x_2^n-1}{x_2-1} + \dfrac{x_2+1}{x_2-x_1}\dfrac{x_1^n-1}{x_1-1},\\
\\
c_n=1-a_n-b_n.
\end{array}
\right.
\end{equation*}
Using for $i=1,2$ and $n\in\mathbb{N}$ the following notation,
\begin{equation}\label{Xi}
X_i(n) = \dfrac{x_i^n-1}{x_i-1}, 
\end{equation}
and since $x_2-x_1=\sqrt{\Delta}$, the system above
can be rewritten as
\begin{equation}
\label{soluce}
\left\{
\begin{array}{l}
a_n = \dfrac{X_2(n)-X_1(n)}{\sqrt{\Delta}},\\
\\
b_n = \dfrac{(x_2+1)X_1(n)-(x_1+1)X_2(n)}{\sqrt{\Delta}},\\
\\
c_n=1+\dfrac{x_1 X_2(n)-x_2 X_1(n)}{\sqrt{\Delta}}.
\end{array}
\right.
\end{equation}
By evaluating \eqref{resol} in $M$ and due to the
theorem of Cayley-Hamilton, we finally have for every integer $n\geqslant 1$,
\begin{equation}
\label{Mn}
M^n = a_n M^2 + b_n M + c_n I_3,
\end{equation}
where $I_3$ is the identity matrix of size 3, 
$a_n, b_n,$ and $c_n$ are given by \eqref{soluce}, and $M^2$ is given by
$$
M^2=\left(
\begin{array}{c|c|c}
a^2 + 2ab + ac - 2a & -a^2 - ab - ac  & -ab + ad - b^2 \\
+ b^2 + be - 2b + 1  & - ad + 2a + bf & - be - bf + 2b\\
\hline
-ac - bc - c^2  & ac + c^2 + 2cd - 2c  & bc - cd - d^2 \\
- cd + 2c + de & + d^2 +df - 2d + 1 & - de - df + 2d\\
\hline
-ae - be + cf  & ae - cf - df  & be +df + e^2 + 2ef \\
- e^2 - ef + 2e & - ef - f^2 + 2f & - 2e + f^2 - 2f + 1
\end{array}
\right)
.$$

\subsubsection{Determination of $M^n$ in particular situations}\label{Mnparticular}

Formulations of \eqref{soluce} only hold for $x_1 \neq x_2$, $x_1 \neq 1$, and
$x_2 \neq 1$. We now investigate these latter cases.

\paragraph{Preliminaries}

Let us firstly remark that, as the mutation matrix M is stochastic, we have
necessarily $0\leqslant a+b \leqslant 1$, $0\leqslant c+d \leqslant 1$, and
$0\leqslant e+f \leqslant 1$. These inequalities imply that 
$s \in [0,3]$. Consequently from the definition of $p$ one can check that 
$p=ad+a(e+f)+b(c+d)+bf+c(e+f)+de \leqslant ad + a + b + bf + c +de \leqslant s$, as each parameter is in $[0,1]$.
To sum up,
\begin{equation}
\label{ineg}
 0 \leqslant p \leqslant s \leqslant 3.
\end{equation}

Suppose now that $\Delta \geqslant 0$. Then \eqref{x1x2} and \eqref{ineg} imply that
\begin{equation}
\label{x1x2eval}
 x_1 = \dfrac{-s+2-\sqrt{\Delta}}{2} \in \left[-2;1\right], x_2 = \dfrac{-s+2+\sqrt{\Delta}}{2} \in \left[-\frac{1}{2};\frac{5}{2}\right]
\end{equation}
Note that, as we deal with a stochastic process, the module of the eigenvalues of $M$ are smaller than 1, so $|x_1|\leqslant 1$ and $|x_2| \leqslant 1$.

\paragraph{Suppose that $x_1=1$} Then $-s=\sqrt{s^2-4p} \Longleftrightarrow s=p=0$.
So $a=b=c=d=e=f=0$, and the mutation matrix is equal to the identity of size 3. 
Conversely, if $a=b=c=d=e=f=0$, then $x_1=1$. 

In that situation, the system does not evolve.

\paragraph{Suppose that $x_2=1$ (and $x_1 \neq 1$)} Then $s=\sqrt{s^2-4p} \Longleftrightarrow p=0$.
In that situation, $x_1=1-s$ and $1$ is root of multiplicity 2 of $\chi_2$, whereas $x_1=1-s$ is its third root. As the case $x_1=1$ has already been regarded, we can consider that $s\neq 0$. Using \eqref{resol}, These facts lead to the following
system:

$$\left\{ 
\begin{array}{rl}
1 & = a_n + b_n + c_n,\\
n & = 2 a_n  + b_n, \\
(1-s)^n & = (1-s)^2 a_n  + (1-s)b_n+c_n.\\
\end{array}
\right.$$

Standard computations then give the following formula:
\begin{equation}
\label{uneSoluce}
\left\{
\begin{array}{l}
 a_n = \dfrac{-1+sn+(1-s)^n}{s^2},\\
 b_n = \dfrac{(3-s) +(s^2-2s)n+(s-3)(1-s)^n}{s},\\
 c_n = \dfrac{(s-1)(2s-1)-s(s-1)^2n-(s^2-3s+1)(1-s)^n}{s^2}. 
\end{array}
\right.
\end{equation}

\paragraph{Case $x_1=x_2 \neq 1$ ($\Delta = 0$)}
Then \eqref{x1x2eval} implies that $x_1 = 1-s/2 \in\left[-\frac{1}{2},1\right)$. From a differentiation of \eqref{resol} one 
deduces that $x_1$ satisfies the following system for every $n\in\mathbb{N}^*$,
$$\left\{ 
\begin{array}{cl}
1 & = a_n + b_n + c_n\\
x_1^n & = a_n x_1^2 + b_n x_1 + c_n\\
n x_1^{n-1} & = 2 a_n x_1 + b_n\\
\end{array}
\right.$$
Standard algebraic computations give, since $x_1\neq 1$,
\begin{equation}\label{soluceparticular}
\left\{
\begin{array}{l}
a_n = n\dfrac{x_1^{n-1}}{x_1-1}-\dfrac{X_1(n)}{x_1-1}\\
\\
b_n = X_1(n)-a_n(x_1+1)\\
\\
c_n=1-a_n-b_n
\end{array}
\right.
\end{equation}
where $X_1(n)$ is defined in \eqref{Xi}.

\subsection{Convergence study}
\subsubsection{Convergence study in the general case}\label{sssec:convgen}

We suppose in this section that $x_1 \neq x_2$, $x_1 \neq 1$, and $x_2 \neq 1$. So
formulations of \eqref{soluce} hold for $a_n, b_n$, and $c_n$. We split the study convergence in several sub-cases, that are the objects of Theorems \ref{th1x}-\ref{thfinalcase}.

\begin{thm}
\label{th1x}
Suppose that $\left|x_1\right|<1$ and
$\left|x_2\right|<1$.
Then the frequencies $P_R(n), P_C(n)$, and
$P_T(n)$ of occurrence at
time $n$ of purines, cytosines, and thymines in the considered gene,
 converge to the following values:
\begin{itemize}
 \item $P_R(n)\longrightarrow \dfrac{ce+cf+de}{p-bf+df}$
 \item $P_C(n)\longrightarrow \dfrac{ae+af+bf}{p-bf+df}$
 \item $P_T(n)\longrightarrow \dfrac{ad+bc+bd}{p-bf+df}$
\end{itemize}
\end{thm}

\begin{proof}
If $|x_1|<1$ and $|x_2|<1$ then $X_i(n) \longrightarrow \dfrac{1}{1-x_i}$ for
$i=1,2$ and so
$$a_n \longrightarrow \dfrac{1}{\sqrt{\Delta}}\left(\dfrac{1}{1-x_2}-\dfrac{1}{1-x_1}\right).$$
Denote by $a_\infty$ this limit. We have 
$$a_\infty = \dfrac{x_2-x_1}{\sqrt{\Delta}(1-x_2)(1-x_1)} = \dfrac{1}{(1-x_2)(1-x_1)} = \dfrac{1}{\dfrac{s+\sqrt{\Delta}}{2}\dfrac{s-\sqrt{\Delta}}{2}},$$
and finally
\begin{equation*}
a_\infty = \dfrac{4}{s^2-\Delta}=\dfrac{1}{p}. 
\end{equation*}
Similarly, $b_n = X_1(n)-a_n (x_1+1)$ satisfies 
\begin{equation*}
b_n\longrightarrow \dfrac{1}{1-x_1} - \dfrac{x_1+1}{p}.
\end{equation*}
The following computations
\begin{gather*}
\dfrac{1}{1-x_1} = \dfrac{2}{s+\sqrt{\Delta}} = \dfrac{2(s-\sqrt{\Delta})}{s^2-\Delta} = \dfrac{s-\sqrt{\Delta}}{2p}, \\
\dfrac{x_1+1}{p} = \dfrac{-s+4-\sqrt{\Delta}}{2p},
\end{gather*}
finally yield
\begin{equation*}
b_\infty = \dfrac{s-2}{p}.
\end{equation*}
So  
\begin{equation*}
c_n \longrightarrow 1-a_\infty-b_\infty = \dfrac{p-s+1}{p},
\end{equation*}
and to sum up, the distribution limit is given by
\begin{equation}
\left\{
\begin{array}{l}
a_\infty = \dfrac{1}{p}\\
\\
b_\infty = \dfrac{s-2}{p}\\
\\
c_\infty=\dfrac{p-s+1}{p}
\end{array}
\right.
\end{equation}
Using the latter values in \eqref{Mn}, we can 
determine the limit of $M^n$, which is
$a_\infty M^2+b_\infty M + c_\infty I_3$.
All computations done, we find the following limit for $M^n$,
\begin{equation*}
\dfrac{1}{p-bf+df}\left(\begin{array}{ccc}
ce+cf+de & ae+af+bf & ad+bc+bd \\            
ce+cf+de & ae+af+bf & ad+bc+bd \\                   
ce+cf+de & ae+af+bf & ad+bc+bd\\                   
                  \end{array}
\right).
\end{equation*}
Using \eqref{Pn}, we can thus finally determine
the limit of $P_n = P_0 M^n$ \linebreak $= 
(P_R(0) ~~ P_C(0) ~~ P_T(0)) M^n $.
\end{proof}

\begin{thm}
Suppose that $\left|x_1\right|=1, x_1\neq 1$, and
$\left|x_2\right|\neq 1$.
Then the evolutionary model is not convergent. More precisely, we have:
\begin{itemize}
\item $P_R(2n) = (a^2 + 2ab + ac - 2a+ b^2 + be - 2b + 1)P_R(0) + (-a^2 - ab - ac - ad + 2a + bf)P_C(0)+(-ab + ad - b^2 - be - bf + 2b) P_T(0) $,
\item $P_R(2n+1) = (1-a-b)P_R(0) + aP_C(0)+bP_T(0) $,
\item $P_C(2n) = (-ac - bc - c^2 - cd + 2c + de )P_R(0) + (ac + c^2 + 2cd - 2c  + d^2 +df - 2d + 1)P_C(0)+(bc - cd - d^2 - de - df + 2d) P_T(0) $,
\item $P_C(2n+1) = cP_R(0) + (1-c-d)P_C(0)+dP_T(0) $,
\item $P_T(2n) = (-ae - be + cf- e^2 - ef + 2e )P_R(0) + (ae - cf - df - ef - f^2 + 2f)P_C(0)+(be +df + e^2 + 2ef- 2e + f^2 - 2f + 1) P_T(0) $,
\item $P_T(2n+1) = eP_R(0) + fP_C(0)+(1-e-f)P_T(0) $,
\end{itemize}
\end{thm}

\begin{proof}
Suppose that $|x_1|=1$ and $|x_2| \neq 1$. Then $x_1,x_2 \in \mathbb{R}$, and so $x_1= 1$ or $x_1 = -1$. The first case
has yet been regarded.

If $x_1 = -1$, then $-s+2-\sqrt{\Delta}=-2$ (due to \eqref{x1x2}). So 
$s=4-\sqrt{\Delta}$, and so $s^2-4p = 4-4s+s^2$. Consequently, $p=s-1$.
But $x_1x_2 = 1-s+p$, so $x_1x_2=0$, which leads to $x_2=0$.
Using \eqref{soluce}, we can thus conclude that 
$a_n=1-\dfrac{(-1)^n-1}{-2} = \dfrac{1+(-1)^n}{2}$. 
So $a_{2n}=1$ and $a_{2n+1}=0$.
Similarly, $b_{2n}=0$ and $b_{2n+1}=1$, and finally $c_n=0, \forall n \in \mathbb{N}$.

These values for $a_n, b_n,$ and $c_n$ lead to the following values for $M^n$:
$$
\left\{
\begin{array}{l}
M^{2n} = M^2\\
M^{2n+1} = M.
\end{array}
\right.
$$
\end{proof}

\begin{rem}
The case $|x_1|\neq 1$ and $|x_2| = 1$ necessarilly implies that $x_2 = 1$, which is in contradiction with the assumptions made in preamble of Section \ref{sssec:convgen}.
\end{rem}

\begin{thm}
If $|x_1|=|x_2|$, but $x_1, x_2 \in \mathbb{C}\setminus \mathbb{R}$, 
then $\left(P_R(n) ~~P_C(n) ~~ P_T(n)\right) =
\left(P_R(0) ~~P_C(0) ~~ P_T(0)\right) \times (a_nM^2+b_nM+c_nI_3 )$, where
\begin{itemize}
\item $a_n = -\dfrac{sin\left(\frac{n\theta}{2}\right) sin\left(\frac{(n-1)\theta}{2}\right)}{sin\left(\frac{\theta}{2}\right) sin(\theta)}$,
\item $b_n = \dfrac{2~sin\left(\frac{n\theta}{2}\right) sin\left(\frac{(n-2)\theta}{2}\right) ~ cos\left(\frac{\theta}{2}\right)}{sin(\theta) ~sin\left(\frac{\theta}{2}\right)}$,
\item $c_n = 1 - \dfrac{sin\left(\frac{n\theta}{2}\right) sin\left(\frac{(n-3)\theta}{2}\right)}{sin(\theta) ~sin\left(\frac{\theta}{2}\right)}$.
\end{itemize}
with $ e^{-i\theta} = x_1$.
\end{thm}

\begin{proof}
Suppose that $|x_1|=|x_2|$, but $x_1, x_2 \in \mathbb{C}\setminus \mathbb{R}$. Then $x_1$ and $x_2$ are complex and conjugate, of the form $x_1 = e^{-i\theta}$, $x_2 = e^{i\theta}$, with $\theta \notequiv 0 [\pi] $. So 
$x_1-x_2 = \sqrt{\Delta} = e^{-i\theta}-e^{i\theta} = -2i ~sin(\theta)$, and
$$
\begin{array}{cl}
a_n & = \dfrac{X_2(n)-X_1(n)}{\sqrt{\Delta}} = \dfrac{X_1(n)-X_2(n)}{2i~sin(\theta)} \\\\
2i~sin(\theta) ~a_n & = \dfrac{e^{-in\theta}-1}{e^{-i\theta}-1} - \dfrac{e^{in\theta}-1}{e^{i\theta}-1}\\\\
& = \dfrac{e^{-in\frac{\theta}{2}}}{e^{-i\frac{\theta}{2}}} \dfrac{e^{-in\frac{\theta}{2}} - e^{in\frac{\theta}{2}}}{e^{-i\frac{\theta}{2}} - e^{i\frac{\theta}{2}}} - \dfrac{e^{in\frac{\theta}{2}}}{e^{i\frac{\theta}{2}}} \dfrac{e^{in\frac{\theta}{2}} - e^{-in\frac{\theta}{2}}}{e^{i\frac{\theta}{2}} - e^{-i\frac{\theta}{2}}}\\\\
 & = e^{-i\frac{(n-1)\theta}{2}} \dfrac{-2i~sin\left(\frac{n\theta}{2}\right)}{-2i~sin\left(\frac{\theta}{2}\right)}
 -
  e^{i\frac{(n-1)\theta}{2}} \dfrac{2i~sin\left(\frac{n\theta}{2}\right)}{2i~sin\left(\frac{\theta}{2}\right)}\\\\
& = \dfrac{sin\left(\frac{n\theta}{2}\right)}{sin\left(\frac{\theta}{2}\right)} \left( e^{-i\frac{(n-1)\theta}{2}} - e^{i\frac{(n-1)\theta}{2}} \right).
\end{array}
$$
Finally, 
$$a_n = -\dfrac{sin\left(\frac{n\theta}{2}\right) sin\left(\frac{(n-1)\theta}{2}\right)}{sin\left(\frac{\theta}{2}\right) sin(\theta)}.  $$
Similarly,

$$
\begin{array}{cl}
\sqrt{\Delta} b_n & = (x_2+1)X_1(n) -(x_1+1) X_2(n)\\\\
-2i ~sin(\theta) b_n & = \left(e^{i\theta}+1\right) e^{-i\frac{(n-1)\theta}{2}} \dfrac{sin\left(\frac{n\theta}{2}\right)}{sin\left(\frac{\theta}{2}\right)}
-
\left(e^{-i\theta}+1\right) e^{i\frac{(n-1)\theta}{2}}\dfrac{sin\left(\frac{n\theta}{2}\right)}{sin\left(\frac{\theta}{2}\right)}\\\\
 & = \dfrac{sin\left(\frac{n\theta}{2}\right)}{sin\left(\frac{\theta}{2}\right)} \left[ e^{-i\frac{(n-3)\theta}{2}} + e^{-i\frac{(n-1)\theta}{2}} - e^{i\frac{(n-3)\theta}{2}} - e^{i\frac{(n-1)\theta}{2}}\right]\\\\
 b_n & = \dfrac{sin\left(\frac{n\theta}{2}\right)}{sin(\theta) sin\left(\frac{\theta}{2}\right)}
  \left( sin\left(\frac{(n-3)\theta}{2}\right) + sin\left(\frac{(n-1)\theta}{2}\right)\right).
\end{array}
$$

and finally,
$$b_n = \dfrac{2~sin\left(\frac{n\theta}{2}\right) sin\left(\frac{(n-2)\theta}{2}\right) ~ cos\left(\frac{\theta}{2}\right)}{sin(\theta) ~sin\left(\frac{\theta}{2}\right)}.$$

As $c_n = 1-a_n-b_n$, we have:
$$c_n = 1 - \dfrac{sin\left(\frac{n\theta}{2}\right) sin\left(\frac{(n-3)\theta}{2}\right)}{sin(\theta) ~sin\left(\frac{\theta}{2}\right)}.$$

\end{proof}

\subsubsection{Convergence study in particular situations}

The case where $x_1=1$ has already been discussed, it implies that
 $a=b=c=d=e=f=0$, and so the system does not evolve. The other particular situations are invastigated in the two following theorems.

\begin{thm}
Suppose that $x_2=1$ and $x_1 \neq 1$ (or equivalently $p=0$). Then the system is well formulated if and only if
$M^2+s(s-2)M-(s-1)^2I_3 \neq 0$. In that situation, we have:
\begin{itemize}
 \item either $s \in ]0,2[$, and so $\left(P_R(n) ~~P_C(n) ~~ P_T(n)\right) \longrightarrow 
\left(P_R(0) ~~P_C(0) ~~ P_T(0)\right) \times \dfrac{1}{s^2}[-M^2+s(3-s)M+(s-1)(2s-1)I_3]$.
 \item or $s=2$, and so $\left(P_R(2n) ~~P_C(2n) ~~ P_T(2n)\right) \longrightarrow 
\left(P_R(0) ~~P_C(0) ~~ P_T(0)\right)$ whereas $\left(P_R(2n+1) ~~P_C(2n+1) ~~ P_T(2n+1)\right) \longrightarrow 
\left(P_R(0) ~~P_C(0) ~~ P_T(0)\right)\times (-2M^2+4M+2I_3)$.
\end{itemize}
\end{thm}

\begin{proof}
Using \eqref{uneSoluce}, we can deduce that $M^n$ is equal to:

$\begin{array}{rcl}
a_nM^2+b_nM+c_nI_3 &=& \dfrac{n}{s}[M^2+s(s-2)M-(s-1)^2I_3]\\\\
&&+\dfrac{1}{s^2}(1-s)^n[M^2+s(s-3)M-(s^2-3s+1)I_3]\\\\
&&+\dfrac{1}{s^2}[-M^2+s(3-s)M+(s-1)(2s-1)I_3].\\\\
\end{array}
$

Several cases can be deduced from this equality.
\begin{itemize}
 \item If $s \in ]0,2[$, then $M^n$ is bounded if and only if 
$M^2+s(s-2)M-(s-1)^2I_3 = 0$. In that condition, 
$M^n \longrightarrow \dfrac{1}{s^2}[-M^2+s(3-s)M+(s-1)(2s-1)I_3]$.
 \item If $s = 2$, then another time $M^n$ is bounded if and only if 
$M^2+s(s-2)M-(s-1)^2I_3 = 0$. In that condition, 
$M^{2n} \longrightarrow I_3$, whereas $M^{2n+1} \longrightarrow 
-2M^2+4M+2I_3$.
 \item Finally, if $s>2$, then as $s=a+b+c+d+e+f$ and $a,b,c,d,e,f\in[0,1]$,
we have necessarily at least three coefficients in $a,b,c,d,e,f$ that are
non zero. So at least one product in $abc, abd, abe, abf, acd, ace, acf, 
ade, adf, aef, bcd$, $bce, bcf, bde, bdf, bef, cde, cdf, cef, def$ is 
strictly positive. This is impossible, as $p=ad+ae+af+bc+bd+bf+ce+cf+de$
is equal to 0.
\end{itemize}
\end{proof}

\begin{thm}\label{thfinalcase}
Suppose that $x_1=x_2 \neq 1$ (or equivalently $s^2 = 4p$).
Then the probabilities $P_R(n), P_C(n)$, and
$P_T(n)$ of occurrence at
time $n$ of a purine, cytosine, and thymine
on the considered
nucleotide, converge to the following values:
\begin{itemize}
 \item $P_R(n) \longrightarrow \dfrac{4}{s^2}(ce+cf+de)$,
 \item $P_C(n) \longrightarrow \dfrac{4}{s^2}(ae+af+bf)$,
 \item $P_T(n) \longrightarrow \dfrac{4}{s^2}(ad+bc+bd)$.
\end{itemize}
\end{thm}

\begin{proof}
In that case $\Delta = 0$, meaning that 
\eqref{soluceparticular} holds.
Since $x_1\in\left[-\frac{1}{2},1\right)$, one gets the following limits,
\begin{gather*}
\lim_{n\to\infty} X_1(n) = - \frac{1}{1-x_1}, \\
\lim_{n\to\infty} x_1^n = 0,
\lim_{n\to\infty} n x_1^{n-1} = 0,
\end{gather*}
and finally $(a_n,b_n,c_n)$ converges to $(a_\infty,b_\infty,c_\infty)$ with
\begin{equation*}
\begin{cases}
a_\infty = \dfrac{1}{(1-x_1)^2} = \dfrac{4}{s^2} \\[8pt]
b_\infty = \dfrac{-2x_1}{(1-x_1)^2} = 4~\dfrac{s-2}{s^2} \\[8pt]
c_\infty = \dfrac{2 x_1-1}{(1-x_1)^2} + 1 = \left(1 -\dfrac{2}{s}\right)^2
\end{cases}
\end{equation*}

Using these values in \eqref{Mn}, we can 
determine the limit of $M^n$, which is
$a_\infty M^2+b_\infty M + c_\infty I_3$, where
$I_3$ is the identity matrix of size 3.
All computations done, we find
$$M^n\longrightarrow \dfrac{4}{s^2}\left(\begin{array}{ccc}
M_{11} & M_{12} & M_{13} \\            
M_{21} & M_{22} & M_{23} \\                   
M_{31} & M_{32} & M_{33} \\                   
                  \end{array}
\right)$$
with
$M_{11} = \dfrac{s^2}{4}-p+ce+cf+de$, $M_{12} = ae+af+bf$, $M_{13} = ad+bc+bd$,
$M_{21} = ce+cf+de$, $M_{22} = \dfrac{s^2}{4}-p+ae+af+bf$, $M_{23} = ad+bd+bc$,
$M_{31} = ce+de+cf$, $M_{32} = ae+af+bf$, and $M_{33} = \dfrac{s^2}{4}-p+ad+bc+bd$.
However, since $x_1=x_2$, we have 
$\Delta = s^2-4p=0$ and so
$$M^n \longrightarrow 
\dfrac{4}{s^2}\left(
\begin{array}{ccc}
ce+cf+de & ae+af+bf & ad+bc+bd\\
ce+cf+de & ae+af+bf & ad+bc+bd\\
ce+cf+de & ae+af+bf & ad+bc+bd
\end{array}
\right),
$$
\end{proof}

\section{Application in Concrete Genomes Prediction}



\label{appnum}

We consider another time the numerical values for mutations published 
in~\cite{Lang08}. Gene \emph{ura3} of the Yeast \textit{Saccharomyces cerevisiae} has a mutation rate of
$3.80 \times 10^{-10}$/bp/generation~\cite{Lang08}. As this gene is constituted by 804
nucleotides, we can deduce that its global mutation rate per generation is equal to
$m = 3.80 \times 10^{-10} \times 804 = 3.0552\times 10^{-7}$.
Let us compute the values of $a,b,c,d,e,$ and $f$.
The first line of the mutation matrix is constituted by 
$1-a-b = P(R \rightarrow R)$, $a=P(R \rightarrow T)$,
and $b=P(R \rightarrow C)$. $P(R \rightarrow R)$ takes into account
the fact that a purine can either be preserved (no mutation, probability $1-m$), 
or mutate into another purine ($A \rightarrow G$, $G \rightarrow A$).
As the generations pass, authors of~\cite{Lang08} have counted 0 mutations of
kind $A \rightarrow G$, and 26 mutations of kind $G \rightarrow A$.
Similarly, there were 28  mutations $G \rightarrow T$ and 8: $A \rightarrow T$,
so 36: $R \rightarrow T$. Finally, 6: $A \rightarrow C$ and 9: $G \rightarrow C$ 
lead
to 15: $R \rightarrow C$ mutations. The total of mutations to consider when 
evaluating the first line is so equal to 77. All these considerations lead
to the fact that $1-a-b=(1-m)+m\dfrac{26}{77}$, $a=\dfrac{36m}{77}$, and $b=\dfrac{15m}{77}$. A similar reasoning leads to $c=\dfrac{19m}{23}$, $d=\dfrac{4m}{23}$, 
$e=\dfrac{51m}{67}$, and $f=\dfrac{16m}{67}$. 

In that situation, $s=a+b+c+d+e+f=\dfrac{205m}{77} \approx 8.134\times 10^{-7}$, and $p= \dfrac{207488 m^2}{118657}
\approx 1.632\times 10^{-13}.$
So $\Delta = s^2-4p = \dfrac{854221 m^2}{9136589} >0$, $x_1 = 1-\dfrac{m}{2}\left(\dfrac{205}{77}+\sqrt{\dfrac{854221}{9136589}}\right)$, and $x_2 = 1-\dfrac{m}{2}\left(\dfrac{205}{77}-\sqrt{\dfrac{854221}{9136589}}\right)$.
As $x_1\approx 0.9999685 \in [0,1]$ and $x_2 \approx 0.9999686 \in [0,1]$, we have, due to Theorem~\ref{th1x}:
\begin{itemize}
 \item $P_R(n)\longrightarrow \dfrac{ce+cf+de}{p-bf+df}$
 \item $P_C(n)\longrightarrow \dfrac{ae+af+bf}{p-bf+df}$
 \item $P_T(n)\longrightarrow \dfrac{ad+bc+bd}{p-bf+df}$
\end{itemize}

Using the data of~\cite{Lang08}, we find that $P_R(0)=\dfrac{460}{804}\approx 0.572$,
$P_C(0)=\dfrac{133}{804}\approx 0.165$, and
$P_T(0)=\dfrac{211}{804}\approx 0.263$. So $P_R(n) \longrightarrow 0.549$, $P_C(n) \longrightarrow 0.292$, and 
$P_T(n) \longrightarrow 0.159$.
Simulations corresponding to this example are given in Fig.~\ref{3param}.

\begin{figure}
\centering
\includegraphics[scale=0.47]{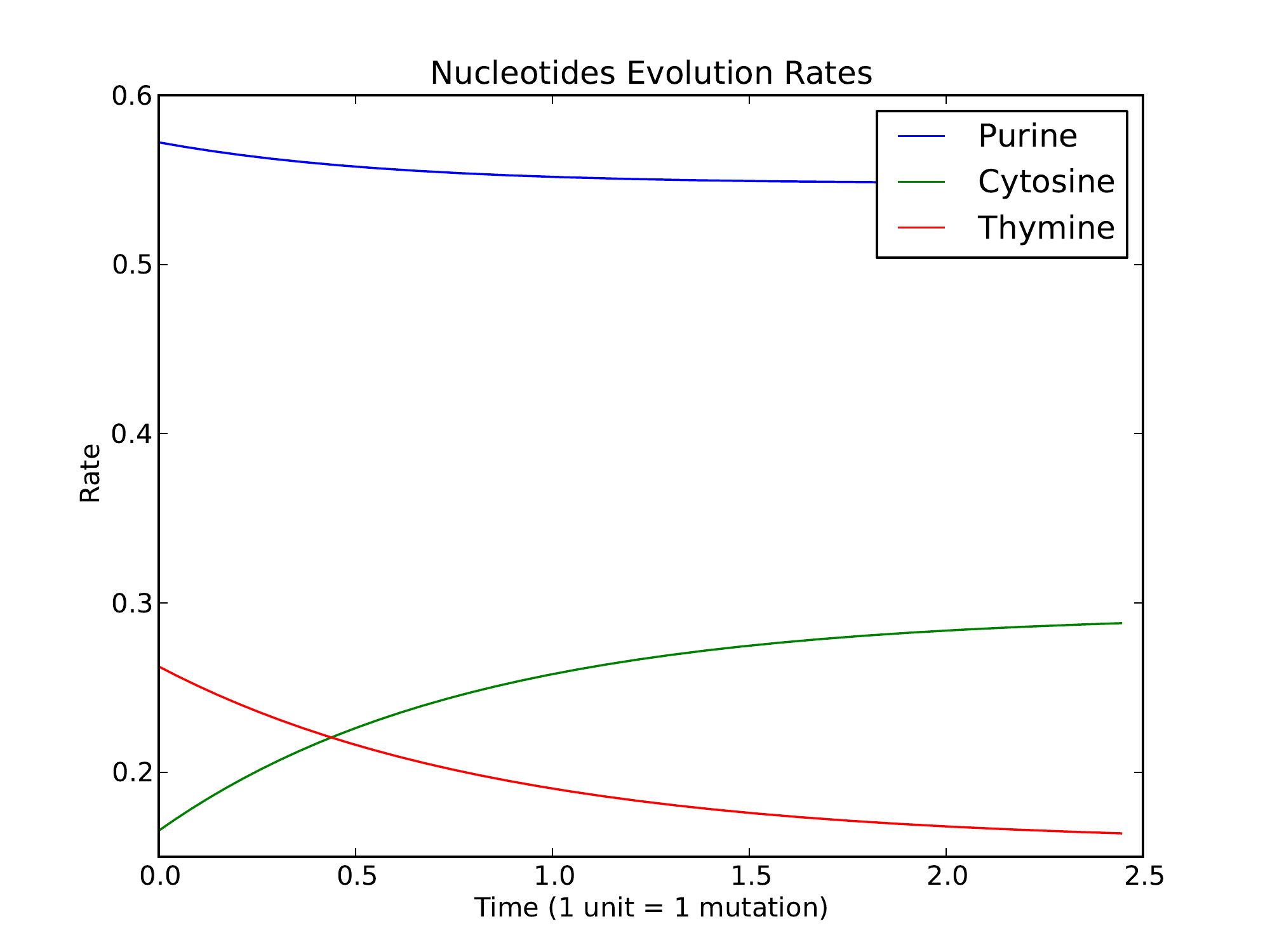}
\caption{\small Prediction of evolution concerning the purine, thymine, and cytosine rates in \emph{ura3}. Non-symmetric Model of size $3\times 3$.}
\label{3param}
\end{figure}

\section{Final Remarks}
\label{Conclusion}

In this document, a formulation of the non symmetric discrete model of size 
$2 \times 2$ has been proposed, which studies a DNA evolution taking into 
account purines and pyrimidines mutation rates. A simulation has been performed, 
to compare the proposal to the well known Jukes and
Cantor model. Then all non-symmetrical models of size 3x3 that have 6 parameters have been studied theoretically.
They have been tested with numerical 
simulations, to make a distinction 
between cytosines and thymines in the former proposal.
These two models still remain generic, and can be adapted to a large panel of
applications, replacing either the couple (purines, pyrimidines) or the tuple
(purines, cytosines, thymines) by any categories of interest.


Remark that the $ura3$ gene is not the unique example of a DNA sequence of interest such that none of
the existing nucleotides evolution models cannot be applied due to a complex
mutation matrix. For instance, a second gene called $can1$ has been studied too
by the authors of~\cite{Lang08}. Similarly to gene $ura3$, usual models cannot be used to predict the
evolution of $can1$, whereas a study following a same canvas than what has 
been proposed in this research work can be realized.

In future work, biological consequences
of the results produces by these models will be systematically investigated. Then,
the most general non symmetric model of size $4 \time 4$ will be regarded in some
particular cases taken from biological case studies, and the possibility of 
mutations non uniformly distributed will then be regarded.
Finally, this $4\times 4$ general case will be investigated using Perron-Frobenius
based approaches instead of using methods directly inspired by linear algebra,
in order to obtain the most global results on mutation matrices.

\bibliographystyle{model1a-num-names}

\end{document}